\documentclass[aip,jmp,longbibliography]{revtex4-1}
\usepackage{amsmath}
\usepackage{amssymb}

\usepackage{color}

\usepackage[normalem]{ulem} 

\usepackage{graphicx,amsfonts}
\usepackage{amsthm}
\usepackage{latexsym}
\usepackage{amsfonts}

\usepackage{hyperref}

\newtheorem{proposition}{Proposition}
\newtheorem{cor}{Corollary}
\newtheorem{lemma}{Lemma}

\renewcommand{\Re}{{\rm Re\,}}

\def\rem[#1]{{\bf [#1]}}

\newcommand{\Z}{\mathbb{Z}} 
\newcommand{\C}{\mathbb{C}} 
\newcommand{\R}{\mathbb{R}} 
\newcommand{\T}{\mathbb{T}} 
\newcommand{\id}{\mathbb{I}} 

\newcommand{\numb}{number } 
\newcommand{\Numb}{Number } 

\def\idty{{\leavevmode\rm 1\mkern -5.4mu I}} 
\def\braket#1#2{\langle #1,#2\rangle}

\def\ketbra #1#2{{\vert#1\rangle\langle#2\vert}}
\def\kettbra#1{\ketbra{#1}{#1}}
\def\abs#1{|#1|}
\def\norm#1{||#1||}

\def\Fou{{\mathcal F}}
\def\HH{{\mathcal H}}
\def\BB{{\mathcal B}}
\let\veps\varepsilon
\def\mby{{\circ}} 
\def\mby{_}

\def\mst{{\rm std}}
\def\mdi{{\rm dis}} 
\def\mar{{\rm arc}}
\def\mch{{\rm cho}}
\def\dst{d_\mst}
\def\ddi{d_\mdi}
\def\dar{d_\mar}
\def\dch{d_\mch}
\def\couples#1#2#3{#1\!\lhd\!#2{\rhd}#3}
\def\PU{\hbox{\textsf{PU}}}
\def\MU{\hbox{\textsf{MU}}}
\def\CU{\hbox{\textsf{CU}}}

\def\tr{\mathop{\rm tr}\nolimits}

\def\Mee{{\mathcal M}(\veps_1,\veps_2)}
\def\essential{\mathop{\rm ess}}

\def\nmax{N_{\rm max}}

\def\hypoo{\,{}_1{\mathrm F}_1}
\def\matze{{\rm ce}_0}

\begin{document}
\title{Sharp uncertainty relations for \numb and angle }

\author{Paul Busch}
\email{paul.busch@york.ac.uk}
\affiliation{Department of Mathematics, University of York, York, United Kingdom}

\author{Jukka Kiukas}
\email{jek20@aber.ac.uk}
\affiliation{Department of Mathematics, Aberystwyth University, SY23 3BZ Aberystwyth, United Kingdom}

\author{R.F. Werner}
	\email{reinhard.werner@itp.uni-hannover.de}
\affiliation{Institut f\"ur Theoretische Physik, Leibniz Universit\"at, Hannover, Germany}

\begin{abstract} We study uncertainty relations for pairs of conjugate variables like \numb and angle, of which one takes integer values and the other takes values on the unit circle. The translation symmetry of the problem in either variable implies that measurement uncertainty and preparation uncertainty coincide quantitatively, and the bounds depend only on the choice of two metrics used to quantify the difference of \numb and angle outputs, respectively. For each type of observable we discuss two natural choices of metric, and discuss the resulting optimal bounds with both numerical and analytic methods. We also develop some simple and explicit (albeit not sharp) lower bounds, using an apparently new method for obtaining certified lower bounds to ground state problems.
\end{abstract}\maketitle

\section{Introduction}



The study of uncertainty relations has experienced  a major boost in recent years. As more and more experiments reach quantum limited accuracy, sharp quantitative uncertainty and error bounds become more relevant. It has also become evident that the subject of quantum uncertainty cannot be reduced to the classic standard uncertainty relation that was first made rigorous by Kennard \cite{Kennard}, Robertson \cite{Robertson} and others and is now found in every textbook. The purpose of this paper is to present a new case study that exemplifies the three main directions of generalization currently being pursued. The example at hand is given by the number-angle pair of observables, where by  ``number'' we understand an observable whose spectrum is the set of all integers.

 The first  extension of the uncertainty principle concerns the set of scenarios to which quantitative uncertainty relations apply. 
 The Kennard relation is a ``preparation uncertainty relation'', i.e., a quantitative expression of the observation that there is no state preparation for which the distributions of two observables under consideration are both sharp. This relation can be tested, in principle, by separate runs of precise measurements of the two observables, performed on  ensembles of systems in the same state.
  In contrast to this, one may consider attempted joint measurements of noncommuting observables, as already intuitively envisaged by Heisenberg \cite{Heisenberg27}; one finds that such joint measurements are constrained by ``measurement uncertainty relations'' \cite{Wer04,BP07,BLWprl} which describe the unavoidable error bounds. 

The number-angle pair also highlights the need to consider uncertainty measures other than standard deviations. 
The second direction of generalization to be considered is thus in the concrete mathematical expressions measuring the ``sharpness'' of distributions, or the error of an approximate measurement. 
We will consider two alternative types of uncertainty measures and associated error measures for each of the two observables concerned.

Finally, the third direction of generalization concerns the forms of preparation and measurement uncertainty relations that are applicable to arbitrary pairs of observables. The Robertson relation involving the expectations of the commutator in the lower bound fails to provide preparation uncertainty relation in the above sense because the only state-independent lower bound one can get from it is zero. Nevertheless, non-trivial state-independent bounds usually do exist. Uncertainty is really a ubiquitous phenomenon, joint measurability or simultaneous sharp preparability are the exceptions rather than the rule. Accordingly, the problem of establishing tight uncertainty relations for pairs of observables amounts to the task of establishing their (preparation or measurement) \emph{uncertainty regions}, defined as the set of pairs of uncertainty values in all states (for preparation uncertainty) and the set of pairs of error values in all possible joint measurements (for measurement uncertainty), respectively.


It is a pleasing property of the case at hand---\numb and angle---that an essentially complete treatment can be given. This is due to the phase space symmetry which implies, exactly as for standard position and momentum \cite{BLWjmp}, that ``metric distance'' and ``calibration distance'' for the error assessment of observables satisfy the same relations and also that they are quantitatively the same as corresponding preparation uncertainty relations. In this paper we deduce the ensuing measurement uncertainty relations for \numb and angle.

Our paper is organized according to the methods employed. In Section \ref{sec:basics} we review the conceptualization of the preparation uncertainty measures and associated error measures to be used throughout the paper. This is followed by an overview of our main results (Section \ref{sec:overview}). In Section \ref{sec:sym} we show how the phase space symmetry can be exploited to find optimal joint measurements among the covariant phase space observables. Here we show the identity of preparation uncertainty regions and measurement uncertainty (or error) regions in the case at hand, introducing calibration uncertainty regions as a mediating construct. The lower boundary of the uncertainty regions is characterized as a ground state problem. Numerical estimates of the optimal tradeoff curves are shown in Section \ref{sec:numerics}.
Next, Sections \ref{sec:exact} and \ref{sec:lower} give determinations of the exact ground states and lower boundaries for the various combinations of deviation measures, and we show how existing uncertainty relations for \numb and phase can be reproduced or strengthened using our systematic approach. We conclude with an outlook in Section \ref{sec:outlook}.

\section{Conceptual Uncertainty Basics}\label{sec:basics}

\subsection{\Numb and angle observables}
The complementarity between \numb and angle appears in physics in various guises. Essentially, this is for any parameter with a natural periodicity. Geometric angles are one case, with the complementary variable given by a component of angular momentum. Another important case is quantum optical phase, which is complementary to an harmonic oscillator Hamiltonian. At least for preparation uncertainty this is no difficulty, since a relation valid for all states also holds for states supported on the subspace of positive integers. One does get additional or sharper relations from building in this constraint, however. A further important field of applications is quasi-momentum with values in the Brillouin-Zone for a lattice system (of which we only consider the one-dimensional case here). This is also related to form factors arising in the discussion of diffraction patterns and fringe contrast from periodic gratings \cite{BBK}.

The literature on angle-angular momentum uncertainty is almost exclusively concerned with the preparation scenario, although the lack of an error disturbance relation has been noted \cite{tanimu}. In the preparation case a major obstacle was the Kennard and Robertson \cite{Robertson} relation and their role as a model how uncertainty relations should be set up. There is nothing wrong with an observable with outcomes on a circle. But much work was wasted on the question of how to represent ``angle'' measurements by a selfadjoint operator \cite{JudgeLewis,Kraus}. Additional unnecessary confusion in the case of semibounded \numb and quantum optical phase was generated by the ignorance or lack of acceptance of generalized (``POVM'') observables. On the positive side, an influential paper by Judge \cite{judge63} produced a relation (for the arc metric), and conjectured an improvement, which was proved shortly afterwards \cite{EvettMahmoud}. In this context the role of ground state problems for finding optimal bounds, which is also the basis of our methods, seems to have appeared for the first time \cite{vanLeuven}. The appearance of the chordal metric grew out of the approach of avoiding the ``angle'' problem, replacing $\theta$ by the two selfadjoint operators $\cos\theta$ and $\sin\theta$. 

\subsection{Measures of uncertainty and error}
Both in preparation and in measurement uncertainty we have to assess the difference of probability distributions: For preparation uncertainty it is the difference from a sharp distribution concentrated on a single point. This is also the basis of calibration error assessment. For metric error we also need to express the distance of two general distributions. We want to express the distance between distributions on the same scale as the distance of points. 
For example, in the case of position and momentum errors, for all error measures one considers $\Delta Q$ to be measured in length units and $\Delta P$ in momentum units, and this is satisfied by the choice of standard deviation for these measures.

So let us assume that the outcomes of some observable (represented by a POVM, a positive operator valued measure) lie in a space $X$ with metric $d$. For real valued quantities like a single component of position or momentum this usually means $X=\R$ and $d(x,y)=|x-y|$.

We now extend the distance function on the points to a distance between a probability measure $\mu$ on $X$ and a point $x\in X$. It will just the be the mean distance from $x$:
\begin{equation}\label{meandist}
 d_\alpha(\mu,x)=\left(\int\!\mu(dy) d(y,x)^\alpha\right)^{\frac1\alpha}.
\end{equation}
Here the exponent $\alpha\in[1,\infty)$ gives some extra flexibility as to how large deviations are weighted relative to small ones. The $\alpha^{\rm th}$ root ensures that the result is still in the same units as $d$, and also that for a point measure $\mu$ concentrated on a point $y$ we have $d_\alpha(\mu,x)=d(y,x)$ for all $\alpha$. It is also true for all $\alpha$ that $d_\alpha(\mu,x)=0$ happens only for the point measure at $x$. We will later mostly choose $\alpha=2$ and drop the index $\alpha$; in this case $d_2(\mu,x)$ is the root mean square  distance from $x$ for points distributed according to $\mu$.

With this measure of deviation of a distribution $\mu$ from the point we can introduce the generalized standard deviation,
\begin{equation}\label{stdev}
  d_\alpha(\mu,*)=\min_{x\in X}d_\alpha(\mu,x).
\end{equation}
Note that for $\alpha=2$, $X=\R$ and $d(x,y)=|x-y|$ this recovers exactly the usual standard deviation, with the minimum being attained at the mean of $\mu$. The symbol $*$ is just a reminder of the minimization, and emphasizes that $d_\alpha(\mu,*)$ is just the distance of $\mu$ from the set of point measures.

For this interpretation to make sense we must also let the second argument of $d_\alpha$ be a general probability distribution $\nu$, resulting in a metric on the set of probability measures.
The canonical definition here is the transport distance \cite{Villani}
\begin{equation}\label{meanmindist}
 d_\alpha(\mu,\nu)=\inf_\gamma\left\lbrace\int\!\gamma(dx\,dy) d(x,y)^\alpha \Bigm\vert \couples\mu\gamma\nu\right\rbrace^{\frac1\alpha},
\end{equation}
where ``$\couples\mu\gamma\nu$'' is a shorthand for $\gamma$, the variable in this infimum, being a ``coupling'' of $\mu$ and $\nu$, i.e., it is a measure on $X\times X$ with $\mu$ and $\nu$ as its marginal distributions. One should think of $\gamma$ as a plan for converting the distribution $\mu$ into $\nu$, maybe for some substance rather than for probability. The cost of transferring a mass unit from $x$ to $y$ is supposed to be $d(x,y)^\alpha$, and the plan $\gamma$ records just how much mass is to be moved from $x$ to $y$. The marginal property means that the initial distribution is $\mu$ and the final one $\nu$. Then $d_\alpha(\mu,\nu)^\alpha$ is the optimized cost. When the final distribution is a point measure, there is not much to plan, and we recover \eqref{meandist}. Therefore there is little danger of confusion in using the same symbol for the metrics of points and of probability measures.

Now we can use these notions of spread and distance for expressing uncertainties related to observables $A,B$ with outcome spaces $X$ and $Y$, each with a suitably chosen metric and error exponent. Let us denote by $\rho\mby A$ the probability measure of outcomes in $X$ upon measuring $A$ on systems prepared according to $\rho$. To express {\it preparation uncertainty}, let us consider the set $\PU$ of generalized variance pairs 
\begin{equation}\label{devpairs}
  \PU=\Bigl\{ \bigl(d_\alpha(\rho\mby A,*)^\alpha,\,d_\beta(\rho\mby B,*)^\beta\bigr)\bigm\vert \rho \text{ a state}\Bigr\}.
\end{equation}
A preparation uncertainty relation is some inequality saying that the uncertainty region does not extend to the origin: the two deviations cannot both be simultaneously small. If this set is known, we consider it as the most comprehensive expression of preparation uncertainty. Its description by inequalities for products or weighted sums or whatever other expression is a matter of mathematical convenience, and we will, of course, develop appropriate expressions. A lower bound for the product is useful {\it only} for position and momentum and its mathematical equivalents. In this case the dilatation invariance $(q,p)\mapsto(\lambda q,\lambda^{-1}q)$ forces the uncertainty region to be bounded by an exact hyperbola. But if one of the observables considered can take discrete values, the set will reach an axis, making every state-independent lower bound on the product trivial. 

We should note that the set \eqref{devpairs} is in general not convex, and can have holes (for examples, see \cite{AMU}). However, in order to express lower bounds, the essence of uncertainty, it makes no difference if we fill in these holes, and include with every point also those for which both coordinates are larger or the same. The resulting set, the `monotone hull'  $\PU^+$ of $\PU$, is bounded below by the graph of a non-increasing function, the {\it tradeoff curve} (see Fig.~\ref{fig:basic4} for examples). It still need not be convex in general, but we will see that convexity holds in the examples we study.

For {\it measurement uncertainty} we again consider two observables (POVMs) $A,B$ with the same outcomes, metrics and error exponents. Now the question is: can $A,B$ be measured jointly? The claim is, usually, that no matter how we try there will be an error in our implementation. So let $A',B'$ be the margins of some joint measurement with outcomes $X\times Y$. Then $A'$ must exhibit some errors relative to $A$, i.e., some output distributions $\rho\mby{A'}$ must be different from $\rho\mby{A}$. We define as the error of $A'$ with respect to $A$ the quantity
\begin{equation}\label{dAA}
  d_\alpha(A',A)=\sup_\rho d_\alpha(\rho\mby{A'},\rho\mby{A}).
\end{equation}
Note that we are using here a worst case quantity with respect to the input state. This is what we should do for a figure of merit for a measuring instrument. If a manufacturer claims that his device $A'$ will produce distributions $\varepsilon$-close to those of $A$ for any input state, he is saying that $d_\alpha(A',A)\leq\varepsilon$. Making such a claim for just a single state is as useless as advertising a clock which measures ``the time $12{:}00$'' very precisely (but maybe no other). Now we can look at the uncertainty region
\begin{equation}\label{MUpairs}
  \MU=\Bigl\{ \bigl(d_\alpha(A',A)^\alpha,\,d_\beta(B',B)^\beta\bigr)\bigm\vert \couples{A'}{\!}{B'}\Bigr\},
\end{equation}
where the 
notation $\couples{A'}{\!}{B'}$ indicates that $A'$ and $B'$ are jointly measurable (that is, they are margins of some POVM that serves as a joint observable). All general remarks made about the preparation uncertainty region $\PU$ also hold for $\MU$.

The supremum in \eqref{dAA} is rather demanding experimentally. Good practice for testing the quality of a measuring device is {\it calibration}, i.e., testing it on states with known properties, and seeing whether the device reproduces these properties. In our case this means testing the device $A'$ on states $\rho$ whose $A$-distribution is sharply concentrated around some $x$, and looking at the spread of the $A'$-distribution around the same $x$. 
We define as the calibration error of $A'$ with respect to $A$ the quantity
\begin{equation}\label{dAAmin}
  \Delta_\alpha^c(A',A)=\lim_{\varepsilon\to0}\sup\left\lbrace d_\alpha(\rho\mby{A'},x)
                           \Bigm\vert d_\alpha(\rho\mby{A},x)\leq\varepsilon\right\rbrace.
\end{equation}
Here the limit exists because the set (and hence the sup) is decreasing as $\varepsilon\to0$. This definition only makes sense if there actually are sufficiently many sharp states for $A$, so we will use this definition only when the reference observable $A$ is projection valued. Since the calibration states in this definition are also contained in the supremum \eqref{dAA}, it is clear that $\Delta_\alpha^c(A',A)\leq d_\alpha(A',A)$, so the set $\CU$ of calibration error pairs will generally be larger than $\MU$.

\section{Setting and Overview of Results}\label{sec:overview}

We now consider systems with Hilbert space $\HH=L^2(\T,d\theta)$, where $\T$ denotes the unit circle, with $d\theta$ the integration over angle. The notation derives from ``torus'' and is customary in group theory. We use it here to emphasize the group structure (of multiplying phases or adding angles mod\,$2\pi$) but also to avoid a fixed coordinatization such as $\T\cong[-\pi,\pi)$, which would misleadingly assign a special role to the cut point $\pm\pi$.  We will refer to $\T$ as our ``position space''. The corresponding ``momentum'' space is $\Z$, and changing to the momentum representation in $\ell^2(\Z)$ is done by the unitary operator of expanding in a Fourier series. With $\{e_n\}_n\in\Z$ denoting the Fourier basis, this means that
\begin{equation}\label{Fourier}
  (\Fou\psi)(n)=\langle{e_n}\vert{\psi}\rangle=\frac1{\sqrt{2\pi}}\int\!\! d\theta\, e^{i n \theta} \psi(\theta).
\end{equation}
We have two natural projection valued observables, the angle (=position, phase) observable $\Theta$ taking values on $\T$, and the \numb (=angular momentum observable) $N$ with values in $\Z$. That is, if $f:\T\to\C$ is some function of the angle variable, $f(\Theta)$ denotes the multiplication operator $(f(\Theta)\psi)(\theta)=f(\theta)\psi(\theta)$, and similarly $g(N)$ denotes the multiplication by a function $g(n)$ in the momentum representation. The outcome proability densities of these observables on an input state $\rho$ are denoted by $\rho\mby\Theta$ and $\rho\mby N$, respectively. Thus,
\begin{eqnarray}
  \int\!\!d\theta\,\rho\mby\Theta(\theta)\,f(\theta)&=&\tr\rho f(\Theta), \\
  \sum_n \rho\mby N(n)\, g(n)&=& \sum_n \langle e_n|\rho|e_n\rangle\ g(n).
\end{eqnarray}
Then the basic claim of preparation uncertainty is that $\rho\mby\Theta$ and $\rho\mby N$ cannot be simultaneously sharp, and the basic claim of measurement uncertainty is that there is is no observable with pairs of outcomes $(\theta,n)$ for which the marginal distributions found on input state $\rho$ are close to $\rho\mby\Theta$ and $\rho\mby N$.

In order to apply the ideas of the previous section, we need to choose a metric in each of these spaces.
For discrete values ($X=\Z$) we can naturally take the standard distance or a discrete metric:
\begin{eqnarray}\label{donZ}
  \dst(n,m)&=&|n-m|    \qquad\text{or}\nonumber\\
  \ddi(n,m)&=&1-\delta_{nm}.
\end{eqnarray}
Similarly, there are two natural choices for angles, depending on whether the basis for the comparison is how far we have to rotate to go from $\theta$ to $\theta'$  (``arc distance'') or else the distance of phase factors   $\exp(i\theta)$ and $\exp(i\theta')$ in the plane (``chordal distance''):
\begin{eqnarray}\label{donT}
  \dar(\theta,\theta')&=&\min_{n\in\Z}|\theta-\theta'-2\pi n|    \qquad\text{or}\nonumber\\
  \dch(\theta,\theta')&=&\left|e^{i\theta}-e^{i\theta'}\right|=2\left|\sin\frac{\theta-\theta'}2\right|.
\end{eqnarray}
The variances based on these bounded metrics will have an upper bound. Since the minimum in \eqref{stdev} makes $d_\alpha(\mu,*)^\alpha$ a concave function of $\mu$, so that we find, by averaging over translates, that the equidistribution has the maximal variance for all translation invariant metrics and all exponents. Both metrics, or rather their quadratic ($\alpha=2$) variances have been discussed before. The $\dar$-variance was used by L\'evy \cite{Levy} and Judge \cite{judge63}, the $\dch$-variance seems to have appeared first in von Mises \cite{vMises}. In fact, the quadratic chordal variance can also be written
as
\begin{eqnarray}\label{vMis}
  d_{\mch,2}(\mu,*)^2&=&\inf_\alpha\int\!\!\mu(d\theta)\left|e^{i\theta}-e^{i\alpha}\right|^2 \nonumber\\
     &=&2\Bigl(1-\sup_\alpha\Re \ e^{-i\alpha} \int\!\!\mu(d\theta)e^{i\theta} \Bigr)\nonumber\\
     &=&2\Bigl(1-\bigl|\langle e^{i\theta}\rangle_\mu\bigr|\Bigr),
\end{eqnarray}
which is von Mises' ``circular variance''. For a review of these choices see Ref.~\onlinecite{breitenberger}. The only property needed in our approach is that the metric should not break the rotation invariance, i.e., it should be a function of the difference of angles. We will therefore use every metric  $d$ also as a single variable function, i.e., $d(x)=d(x,0)$ and $d(x,y)=d(x-y)$. Functions which do not come from a metric have been considered in Ref.~\onlinecite{alonso}.

\begin{figure*}[t]
	\includegraphics[width=0.75\textwidth]{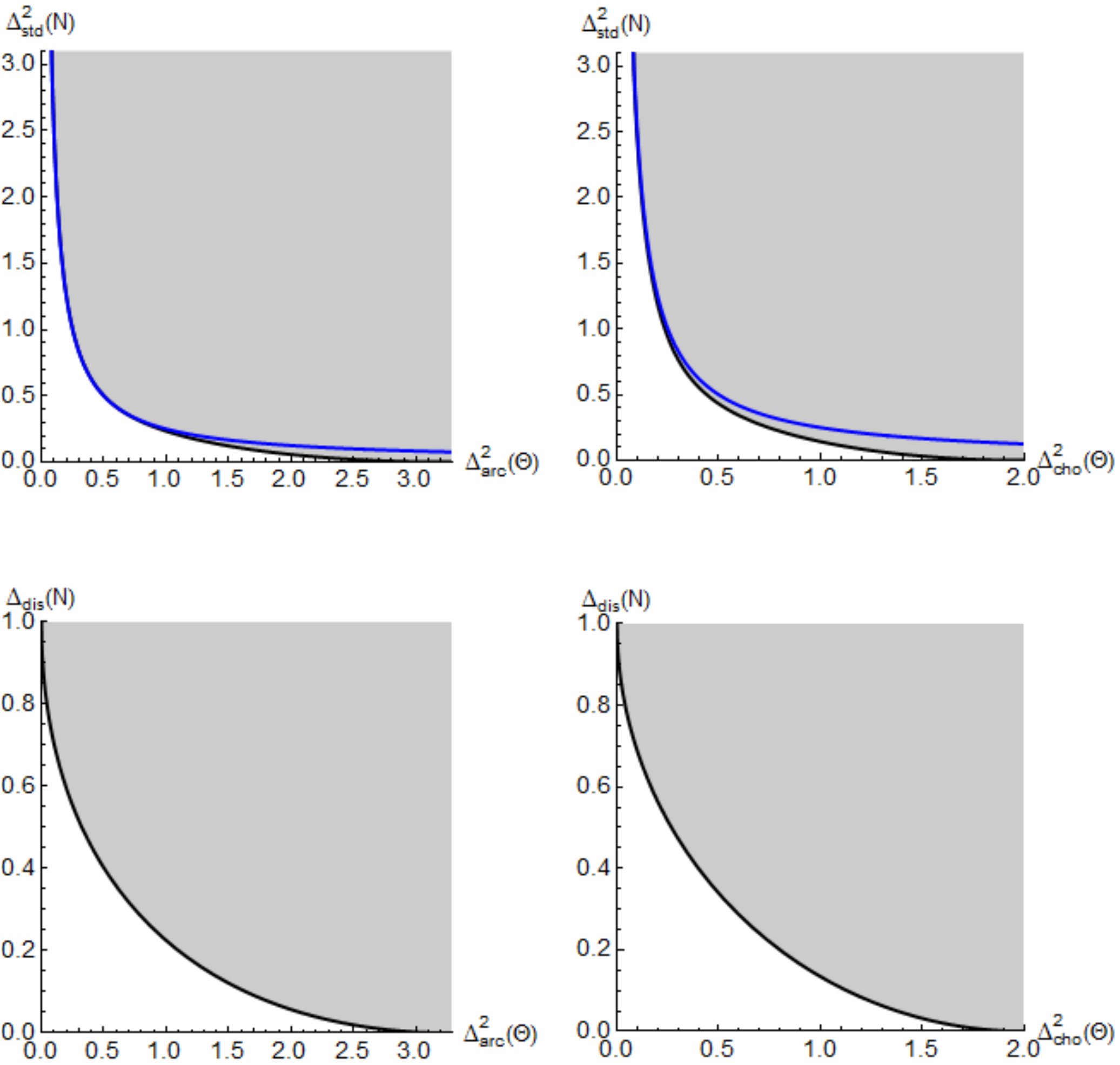}
	\caption{The uncertainty regions for the four pairs of metrics given by \eqref{donZ} and \eqref{donT}. Each region represents at the same time preparation and measurement uncertainty
       (for metric or calibration error criterion).  For the standard metric on $\Z$ the position-momentum bound is shown for comparison as a blue line.
       $\Delta^2_\mst(N)$ is an unbounded quantity. For the others the maximal interval is displayed.}
	\label{fig:basic4}
\end{figure*}

Then we have the following result.

\begin{proposition}\label{mainthm}
For all error exponents and choices of translation invariant metric the three uncertainty regions $\PU^+=\MU^+=\CU^+$ coincide.
They are depicted for $\alpha=\beta=2$ in Figure~\ref{fig:basic4}.
Every point on one of the tradeoff curves belongs to a unique pure state
(resp. a unique extremal phase space covariant joint measurement).
\end{proposition}

The proof of this Proposition is based entirely on the corresponding proof for standard position and momentum \cite{BLWjmp}. We will sketch the main steps in the next section, and also show how the computation of the tradeoff curve can be reduced to solving ground state problems for certain Hamiltonians. The detailed features of these diagrams are then developed in the subsequent sections, sorted by the methods employed. The tradeoff curves  in Figure~1 for the generalized variances, denoted simply $\Delta^2(N),\Delta^2(\Theta)$, 
are determined numerically  (see Section~\ref{sec:numerics}). Since the algorithms employed provide optimal bounds, the figures are correct within pixel accuracy (which can be easily pushed to high accuracy). In fact, the problem is very stable, in the sense that near minimal uncertainty implies that the state (or joint observable) is close to the minimizing one. The bounds for this are in terms of the spectral gap of the Hamiltonian and are also discussed in Section~\ref{sec:numerics}.

However, the only case in which the exact tradeoff curves can be described in closed form is (see Section~\ref{sec:discrete})
\begin{equation}\label{dischord}
  \Delta_\mdi(N)\geq 1-\frac12\sqrt{\Delta^2_\mch(\Theta)\bigl(4-\Delta^2_\mch(\Theta)\bigr)}\ ,
\end{equation}
even though, in all cases, the optimizing states can be expressed explicitly in terms of standard special functions (see Section~\ref{sec:exact}). Therefore simple and explicit lower bounds are of interest. A problem here is that computing the variances for some particular state always produces a point inside the shaded area, i.e., an upper bound to the lower bound represented by the tradeoff curve. This is useless for applications, so in Section~\ref{sec:lower} we develop a procedure proving lower bounds, and thus correct (if suboptimal) uncertainty relations.

\section{Phase Space Symmetry and Reduction to a Ground State Problem}\label{sec:sym}

In this section we briefly sketch the arguments leading to the equality of preparation und measurement uncertainty regions. The full proof is directly parallel to the one given in Ref.~\onlinecite{BLWjmp} for position and momentum. The basic reference for phase space quantum mechanics is Ref.~\onlinecite{QHA}. The theory there is developed for phase spaces of the form $\R^n\times\R^n$, but all results we need here immediately carry over to the general case $X\times\widehat X$, where $X$ is a locally compact abelian group and $\widehat X$ is its dual, in our case $X=\T$ and $\widehat X=\Z$. A systematic extension of Ref.~\onlinecite{QHA} to the general case, including the finer points, is in preparation in collaboration with Jussi Schultz.

\subsection{Covariant phase space observables}

The phase space in our setting is the group $\Omega=\T\times\Z$. We join the position translations and the momentum translations to phase space translations, which are represented by the displacement or {\it Weyl operators}
\begin{equation}\label{Weyl}
  (W(n,\theta)\psi)(x)= e^{-\frac 12 n\theta}e^{-inx}\psi(x-\theta).
\end{equation}
These operators commute up to a phase, so that the operators $T_\omega(A)=W(\omega)^*AW(\Omega)$, i.e., the action of the Weyl operators on bounded operators $A\in\BB(\HH)$, is a representation of the abelian group $\Omega$. A crucial property is the square integrability of the matrix elements of the Weyl operators, which we will use in the form that for any two trace class operators $\rho,\sigma$ on $\HH$ the formula
\begin{equation}\label{sqint}
  \int d\omega \tr\bigl( \rho\,T_\omega(\sigma)\bigr)=\tr \rho\,\tr \sigma,
\end{equation}
where $\int\!\!d\omega=\sum_n\int\!\frac{d\theta}{2\pi}$. Hence, when both $\rho$ and $\sigma$ are density operators the integrand in this equation is a probability density on $\Omega$. 
In summary \cite{QHA}:
\begin{proposition}
Every density operator $\sigma$ serves an observable $F_\sigma$ with outcomes $\omega$ in the phase space $\Omega$, via
\begin{equation}\label{covobs}
  F_\sigma(A)=\int_{\omega\in A}\mskip-20mu d\omega\ T_\omega(\sigma).
\end{equation}
Here $\sigma$ figures as the operator valued Radon-Nikodym density with respect to $d\omega$ of $F_\sigma$ at the origin, and by translation also at arbitrary points.
The observables obtained in this way are precisely the {\it covariant} ones, i.e., those satisfying the equation $T_\omega(F(A))=F(A-\omega)$.
\end{proposition}

For the discussion of uncertainties we need the margins of such observables. One can guess their form from a fruitful classical analogy \cite{QHA}, by which the integrand in \eqref{sqint} can be read as a convolution of $\rho$ and $\sigma$. For classical probability densities on a cartesian product it is easily checked that the margin of the convolution is the convolution of the margins. The same is true for operators, only that the margins of a density operator are  the classical distributions $\rho\mby\Theta$ and $\rho\mby N$ described above.
If $p$ is a probability density on phase space we will also denote by $p\mby\Theta$ and $p\mby N$ the respective margins on $\T$ and $\Z$. In particular, for the output distribution of the covariant observable
$F_\sigma$, i.e., $p(\theta,n)=\tr\rho W(\theta,n)^*\sigma W(\theta,n)$, One checks readily checks the marginal relations
\begin{align}
  p\mby N&= \rho\mby N \ast \sigma\mby N,  \\
  p\mby\Theta&= \rho\mby\Theta \ast \sigma\mby\Theta,
\end{align}
where ``$\ast$'' means the convolution of probability densities on $\Z$ and $\T$. Since this is the operation associated with the sum of independent random variables we arrive at the following principle:
\begin{cor}
  {Both margins of a covariant phase space measurement $F_\sigma$ can be simulated by first making the corresponding ideal measurement on the input state, and then adding some independent noise, which is also independent of the input state. The distribution of this noise is the corresponding margin of the  density operator $\sigma$}.
\end{cor}
This principle is responsible for the remarkable equality of the preparation uncertainty region $\PU^+$ and the measurement uncertainty regions $\MU^+$ and $\CU^+$. Indeed the added noise is what distinguishes the margins of an attempted joint measurement from an ideal measurement, and this has precisely the distribution relevant for preparation uncertainty for $\sigma$.

\subsection{$\MU$ and $\CU$: Reduction to the covariant case}
While this principle explains quite well what happens in the case of covariant observables $F_\sigma$, Proposition~\ref{mainthm} makes no covariance assumption. The key for reducing the general case is the observation that our quality criteria in terms of $d(\Theta',\Theta)$  do not single out a point in phase space. Thus, let  $\Mee$ be the set of observables whose angle margin $F_\T$ is $\veps_1$-close to the ideal observable, $d(F_\T,\Theta)\leq\veps_1$, and whose \numb margin satisfies $d(F_\Z,N)\leq\veps_2$; then this set  is closed under phase space shifts, i.e., is unchanged when we replace $F$ by $F'$ where
\begin{equation}\label{Fshift}
  F'(A)=T_\omega(F(A+\omega))
\end{equation}
Note that the fixed points of all these transformations are precisely the covariant observables. The second point to note is that the set
$\Mee$ is convex, because the  worst case error of an average in a convex combination is smaller than the average of the worst case errors. It is also a compact set in a suitable weak topology. This is in contrast to the set of all observables: Since there are arbitrarily large shifts in $\Z$, we can shift an observable to infinity such that the probabilities for all fixed finite regions go to zero. The weak limit of such observables would be zero or, in an alternative formulation, would acquire some weight on points at infinity (a compactification of $\Z$). For such a sequence, however, the errors would also diverge. It is shown in Ref.~\onlinecite{BLWjmp} that this suffices to ensure the compactness of $\Mee$. Then the Markov-Kakutani Fixed Point Theorem (Ref.~\onlinecite[Thm.~V.10.6]{Markak}) ensures that $\Mee$ contains a common fixed point of all the transformations \eqref{Fshift}.

In summary:
\begin{proposition}
 For every joint observable $F$ on the phase space $\Omega$ there is a covariant one for which the errors are at least as good. 
 \end{proposition}
Therefore for determining $\MU^+$ we can just assume the observable in question to be covariant, implying the very simple form of the margins described above. 
The argument for $\CU^+$ is the same.

\subsection{The post-processing Lemma: $\CU^+=\MU^+$}
We have noted that, in general, $\Delta_\alpha^c(A',A)\leq d_\alpha(A',A)$, because for calibration the worst case analysis is done over a much smaller set of states. Indeed, it is easy to construct examples of observable pairs where the inequality is strict.  There is a general result, Ref.~\onlinecite[Lemma~8]{AMU}, however, which implies equality. The condition is that $A'$ arises from $A$ by classical, possibly stochastic post-processing. That is, we can simulate $A'$ by first measuring $A$, and then adding noise or, in other words, generating a random output by a process which may depend on the measured $A$-value. The noise is described then by a transition probability kernel $P(x,dy)$ for turning the measured value $x$ into somewhere in the set $dy\subset X$. $P$ is thus a description of the noise, and its relevant size is given by the formula
\begin{eqnarray}\label{postpro}
 \Delta_\alpha^c(A',A)&=& d_\alpha(A',A)
                      = \left(\essential_A\text{-}\sup_{x\in X}\int P(x,dy)\ d(x,y)^\alpha \right)^{1/\alpha}  .
\end{eqnarray}
Here the $A$-essential supremum of a measurable function is the supremum of all $\lambda$ such that the level set $\{x|f(x)\geq\lambda\}$ has non-zero measure with respect to $A$. This is needed to ensure that $P(x,dy)$ enters this formula only for values $x$ that can actually occur as outputs of $A$.

\subsection{The covariant case: $\CU^+=\MU^+=\PU^+$}
In the case at hand, the noise is independent of $x$, i.e., $P$ is translation invariant and so is the metric. Therefore, the integral in \eqref{postpro} is simply independent of $x$. Moreover, we know the distribution of the noise on each margin to be the respective margin of $\sigma$, so that the integral is just the $\alpha$-power deviation of the margin from zero. So we get, for any choice of exponents and translation invariant metrics on $\T$ and $\Z$:
\begin{equation}\label{3delta}
 \Delta_\alpha^c(\Theta',\Theta)= d_\alpha(\Theta',\Theta) = d_\alpha(\sigma\mby\Theta,0),
\end{equation}
and similarly for $N$. Note that the last term here is {\it not} the variance $d_\alpha(\sigma\mby\Theta,*)$, because the minimization over $x$ in \eqref{stdev} is missing. Indeed, if $\sigma\mby\Theta$ just had zero variance, i.e, it were a point measure at some point $x\neq0$, we would get a constant shift of size $d(x)$ between the distributions $\rho\mby{\Theta'}$ and $\rho\mby\Theta$, and this would be the errors on the left hand side. So for a fixed $\sigma$ we can only say that the terms in \eqref{3delta} are $\geq d_\alpha(\sigma\mby\Theta,*)$. On the other hand, we are looking for {\it optimal} $\sigma$ and these will be obtained by shifting $\sigma$ in such a way that   \eqref{3delta} is minimized. Hence, as far as uncertainty diagrams are concerned, we can replace the last term by the $\alpha$-deviation. This concludes the proof that the three uncertainty diagrams coincide.

\subsection{Minimizing the variances}\label{sec:minvariance}
We now describe the general method to find the tradeoff curve. The idea is to fix some negative slope $-t$ in the diagram, and ask for the lowest straight line with that slope intersecting $\PU^+$. That is, we look at the optimal lower bound $c()$ such that
\begin{equation}\label{leglower}
  y+t x\geq c(t)\qquad \text{for all}\ (x,y)\in\PU.
\end{equation}
Now both coordinates $x=d_\alpha(\sigma\mby\Theta,0)^\alpha$ and $y=d_\beta(\sigma\mby N,0)^\beta$ are linear functions of $\sigma$, so that the left hand side of \eqref{leglower} is just the $\sigma$-expectation of some operator, namely
\begin{eqnarray}
  y+t x&=& \tr(\sigma H(t))\nonumber\\
          H(t)&=& d_\Z(N)^\beta + t d_\T(\Theta)^\alpha .\label{Ht}
\end{eqnarray}
Here $d_\T$ and $d_\Z$ are the metrics chosen for these spaces, and we used the notation of writing $f(\Theta)$ for the multiplication operator by $f(\theta)$, and its Fourier transformed counterpart, and also the convention that $d(x)=d(x,0)$ for a translation invariant metric. The optimal constant $c(t)$ is thus the lowest expectation $\inf H=\inf\tr(\sigma H)$, i.e., its {\it ground state energy}. Note that for standard position and momentum phase space and $\alpha=\beta=2$ we get here $H=P^2+t Q^2$, a harmonic oscillator, and the well-known connection between its ground state and minimum uncertainty.

We will look into these ground state problems later and for now note some general features.
\begin{enumerate}
\item[(i)] The variable $t$ is positive because we are looking for lower bounds on $x$ and $y$ only. This corresponds in part to taking the monotone closure, and is the reason why we replace $\PU^+$ by $\PU$ in \eqref{leglower}.
\item[(ii)] The best bound on $\PU$ obtained in this way is achieved by optimizing over $t$, i.e., $$y\geq\sup_t\{-tx+c(t)\},$$ the Legendre transform of $c$. This is automatically convex. In other words, the method does not describe $\PU^+$ in general, but its convex hull (the intersection of all half spaces with positive normal containing $\PU$).
\item[(iii)] There may be points on the tradeoff curve for the convex hull which do not really correspond to a realizable pair of uncertainties. However, if we take the collection of ($t$-dependent) ground states, and their variance pairs trace out a continuous curve, we know that the tradeoff curves are the same and the set $\PU^+$ is actually convex and fully characterized by the ground state method.
\end{enumerate}

\begin{figure}[ht]
\includegraphics[width=0.4\textwidth]{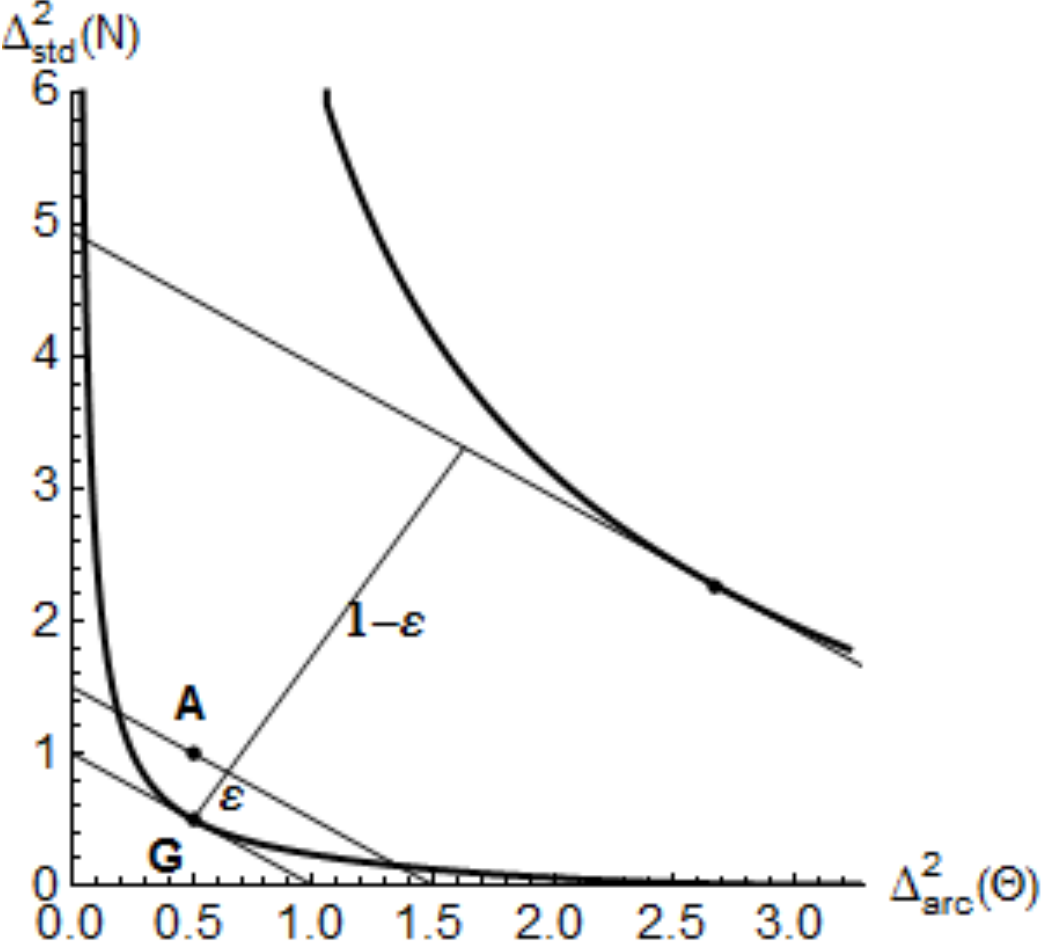}
\caption{Applying the estimate \eqref{nearmin} in an uncertainty diagram. The lower curve is the tradeoff curve obtained from the ground state problem. The upper curve is formed by the uncertainty pairs of the respective first excited states. Suppose the uncertainty pair {\sf A} has been found. We want to show that the corresponding state must be close to the minimal uncertainty state corresponding the point {\sf G}. Drawing the tangent to the tradeoff curve at {\sf G} and the parallel tangent to the upper curve we find from the diagram that the fidelity of the given state to {\sf G} must be at least $1-\varepsilon$.}
\label{fig:nearmin}
\end{figure}

When the ground state problem for $H(t)$ has a gap, it is known that any state with expectations close to the ground state energy must actually be close to the ground state. More precisely, suppose that $H$ has a unique ground state vector, $H\psi=E_0\psi$, and that the next largest eigenvalue is $E_1>E_0$. Then $H\geq E_1\idty- (E_1-E_0)\kettbra\psi$. Now let $E_0<\braket{\phi}{H\phi}=E_\phi<E_1$ for some unit vector $\phi$.  Then by taking the $\phi$-expectation of the operator inequality, we get
\begin{equation}\label{nearmin}
  \abs{\braket\phi\psi}^2\geq \frac{E_1-E_\phi}{E_1-E_0}. 
\end{equation}
In particular, when $E_\phi\approx E_0$, $\phi$ must be close in norm to $\psi$. We can directly apply this principle to the above ground state problems. The basic geometry is described in Fig.~\ref{fig:nearmin}. This shows that the curve of minimizers is continuous. It will also be useful in showing explicitly that the minimizers for different choices of metrics are sometimes quite close to each other, or that some simple ansatz for the minimizer is quantitatively good.

\section{Numerics in truncated Fourier basis}\label{sec:numerics}

Here we only consider the case $d_\Z=\dst$ because for the discrete metric on $\Z$ the ground state problem has an elementary explicit solution (see Section~\ref{sec:discrete}).
The numerical treatment is easiest in the Fourier basis, or rather in the even eigenspace of the \numb operator $N$, for $\abs N\leq\nmax$. The matrix elements of the relevant Hamiltonians for basis vectors in this range can be written down as simple explicit expressions. From these the numerical version of the Hamiltonian is determined as floating point matrix of the desired precision, for which the ground state and first excited state are determined by standard algorithms. All these steps were carried out in Mathematica. The criterion for the choice of $\nmax$ was that the highest-$n$ components of the eigenvectors found should be negligible at the target accuracy. The target accuracy was mostly $5$ digits with computations done in machine precision with $\nmax=80$, but was chosen larger for getting a reliable estimate of the separation of the different state families.

All computations must be considered elementary and highly efficient, even at high accuracy. None of the diagrams in this paper takes computation time longer than a keystroke. It is therefore hardly of numerical advantage to implement the analytic  solutions of Section~\ref{sec:exact}, not in computation time and even less in programming and verification time.

Perhaps the only surprise in this problem is that for the two different metrics $\dar$ and $\dch$ the minimizing state families are so close. Since $V_\mch\leq V_\mar$,  the ground state problems for $H_\mch$ and $H_\mar$ are related. Perturbatively one sees that the ground state energies are indeed similar, up to the expectation of $t(V_\mar(\theta)-V_\mch(\theta))$. The stability statement at the end of Section~\ref{sec:minvariance} then implies that the corresponding ground states are also similar. However, direct comparison gives a norm bound, which is rather better than these arguments indicate:
\begin{equation}\label{normdiff}
  \norm{\psi_\mar(t)-\psi_\mch(t)}\leq 0.145
\end{equation}
for all $t$, corresponding to a fidelity $\geq .98$. This still does not quite reflect the similarity of these two state families: When we allow the $t$-arguments to differ, we get a much better approximation. To make this precise consider the orbits $\Omega_\mar=\{e^{i\alpha}\psi_\mar(t)| t>0, \alpha\in\R\}$, and an analogously defined $\Omega_\mch$. For sets in Hilbert space we use the Hausdorff metric, so that $d_H(\Omega_1,\Omega_2)<\veps$ means that for every point in one set there is an $\veps$-close one in the other. Then one easily gets
\begin{equation}\label{dHOm}
  d_H\bigl(\Omega_\mar,\Omega_\mch\bigr)\leq .028\ 
\end{equation}

Consequently, there is really only one diagram representing the family of minimal uncertainty states, which we show in Fig.~\ref{fig:states}.

\begin{figure}[h]
\includegraphics[width=0.45\textwidth]{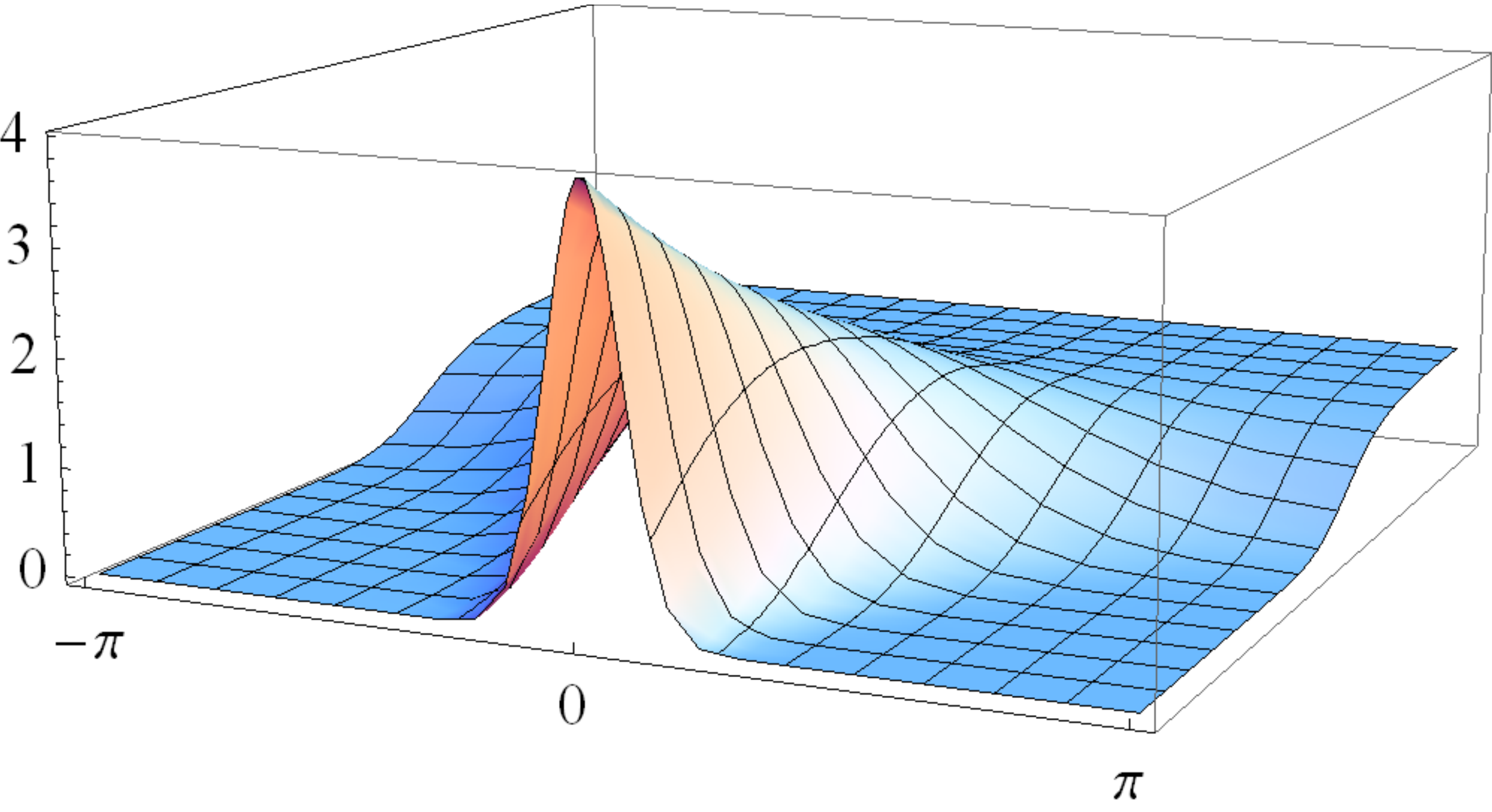}
\includegraphics[width=0.45\textwidth]{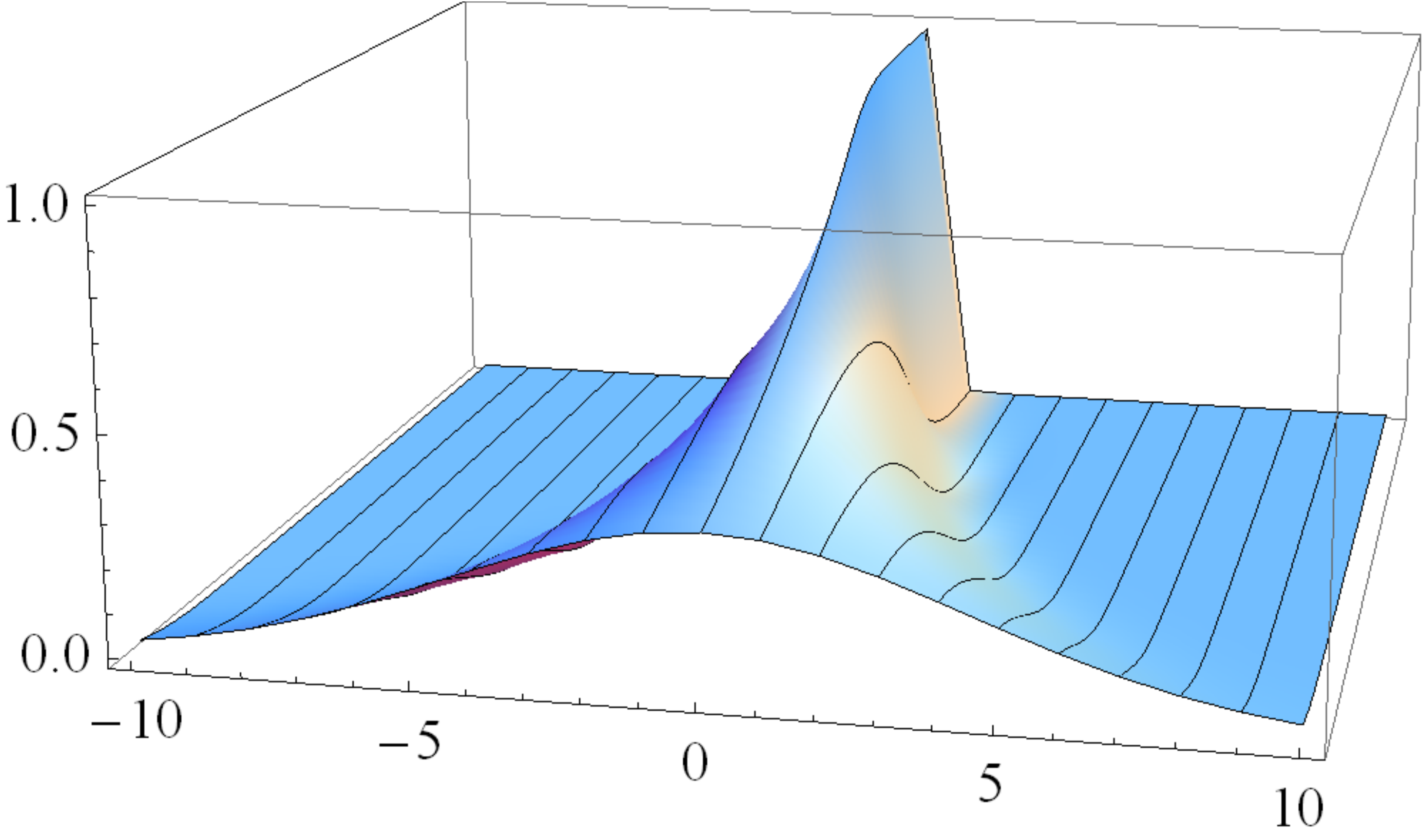}
\caption{Minimum uncertainty wave functions (left) and Fourier coefficients (right) for the case $(\dch,\dst)$ as a function of the family parameter $t$ from \eqref{Ht}, plotted as the depth dimension.
Note that the surface in the diagram on the right is only an aid for better 3D visualization --- only the embedded lines have real significance.}
\label{fig:states}
\end{figure}

\section{Exact ground states}\label{sec:exact}
\subsection{Schr\"odinger operator case}
With $d_{\mathbb Z}=d_{\rm std}$ and $\beta=2$, the ground state problem becomes an instance of the Schr\"odinger operator eigenvalue problem. In fact, writing $V=d_{\mathbb Z}^\alpha(\Theta)$, the optimal constant $c(t)$ for a given $t$ is the smallest value of $\lambda$ such that the differential equation
\begin{equation}\label{diffeq}
-\psi''(\theta) + tV (\theta)\psi(\theta) = \lambda \psi(\theta)
\end{equation}
has a solution on $[-\pi, \pi]$ satisfying the boundary conditions $\psi(-\pi) = \psi(\pi) = \psi'(-\pi) = \psi(\pi) = 0$. By the general theory, we know that the (unique) solution $\psi = \psi_\alpha$ has no zeros, can be chosen to depend smoothly on $t$ (by perturbation theory), and can be chosen to be positive and even (by parity invariance).

Hence, we are in fact looking for an even solution $\psi$ of \eqref{diffeq} with $\psi'(\pi)=0$. At $t = 0$ this vector is just $\psi_\alpha(\theta) = 1/\sqrt{2\pi}$, i.e. a constant. For the two choices $\dar$ and $\dch$, the solutions $\psi_\alpha$ are known special functions; we now proceed to describe them in some detail.

\subsubsection{$V_{\mathbb Z}(\theta)=\dar(\theta)^2=\theta^2$}
The general even solution of \eqref{diffeq} is given by a hypergeometric function \cite{abramowitz}:
\begin{equation}\label{arcpsi}
\psi(\theta)=N^{-1}e^{-\frac{1}{2}\sqrt{t }\,\theta^2}  
             \hypoo\left(\frac{1}{4} \left(1-\frac{\lambda}{\sqrt{t }}\right);\frac{1}{2};\sqrt{t } \theta^2\right),
\end{equation}
where $N$ is the normalisation factor. The boundary condition $\psi'(\pi)=0$ now picks out the eigenvalues $\lambda$ for every $t$ (see Fig.~\ref{fig:arcpsiprime}), of which the smallest is the desired constant $c(t)$. The condition can be expressed by using the standard differentiation formulas for the hypergeometric functions:
\begin{align*}
\left(1-\frac{\lambda}{\sqrt{t }}\right) &\, \hypoo\left(1+\frac{1}{4} \left(1-\frac{\lambda}{\sqrt{t }}\right);\frac{3}{2};\pi ^2 \sqrt{t }\right)
=\hypoo\left(\frac{1}{4} \left(1-\frac{\lambda}{\sqrt{t }}\right);\frac{1}{2};\pi ^2 \sqrt{t }\right)
\end{align*}
However, as far as we could see, the theory of hypergeometric functions seems to offer little help for solving it, or for evaluating the normalization constant $N$ or the Fourier coefficients. Perhaps an elementary expression for $c(t)$ is too much to hope for, since already in much simpler problems, e.g., a particle in a box, where the pertinent transcendental equations involve only trigonometric and linear functions, no ``explicit'' solution can be given either.
\begin{figure}[h]
\includegraphics[width=0.4\textwidth]{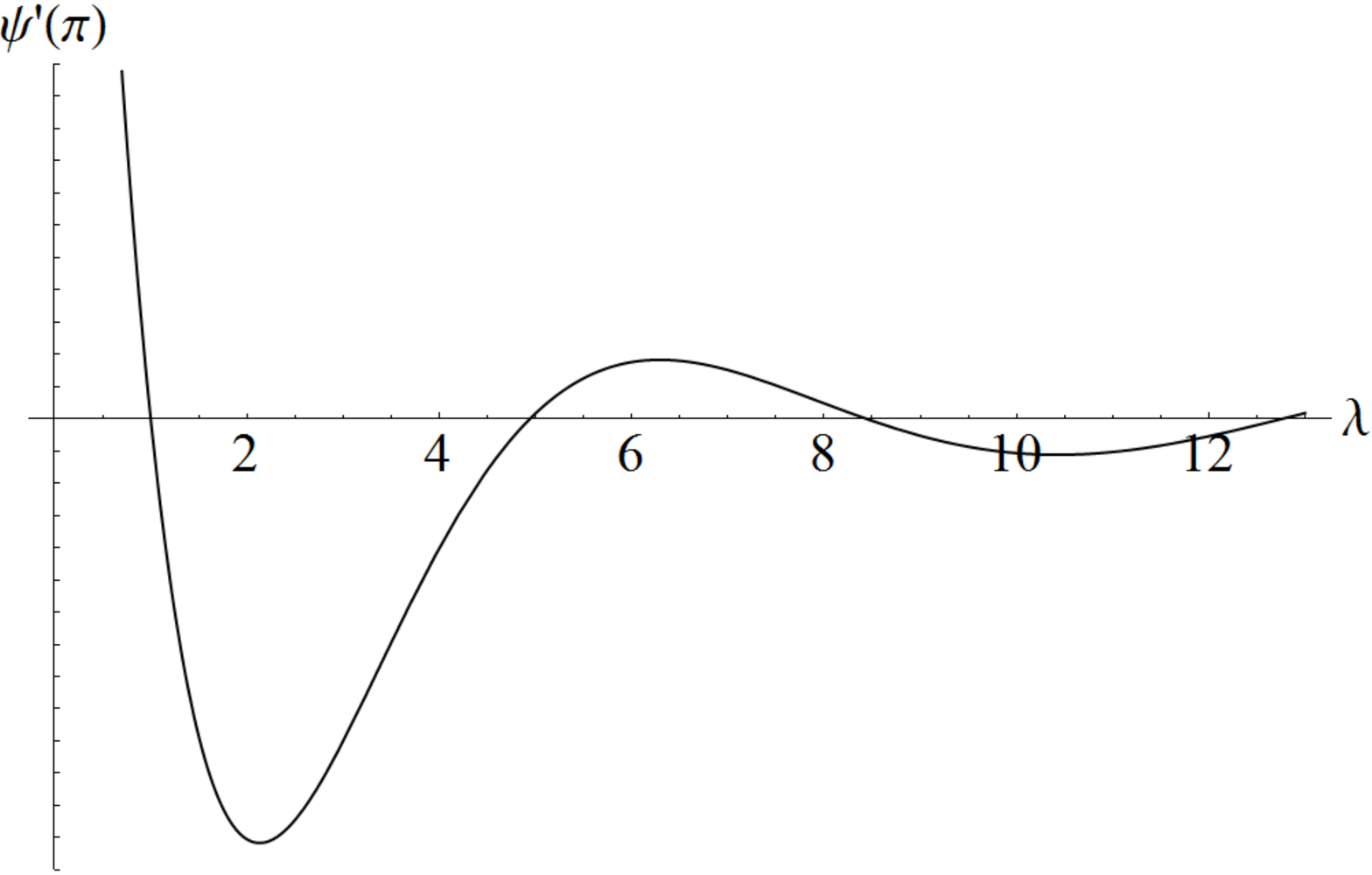}
\caption{The derivative $\psi'(\pi)$ after \eqref{arcpsi}, for $t=1$, as a function of the eigenvalue parameter $\lambda$. The zeros of this function determine the eigenvalues, of which the lowest gives the constant $c(1)$.}
\label{fig:arcpsiprime}
\end{figure}

\subsubsection{$V(\theta)=\dch(\theta)^2=2(1-\cos\theta)$}
In the case of $\dch$ (distance through the circle), the equation \eqref{diffeq} is just the Mathieu equation up to scaling $\theta\mapsto \theta/2$. The even periodic solutions correspond to $-4(2t-\lambda) ={\mathrm a}_r(-4t)$, where the ${\mathrm a}_r(q)$, $r=0,1,\ldots$ are called Mathieu characteristic values \cite{abramowitz}. Our ground state eigenvalues are therefore
$$
c(t) =2t +a_t/4,
$$
where we have used the shorthand $a_t={\mathrm a}_0(-4t)$. Since ${\mathrm a}_0$ is implemented in e.g. Mathematica, we can easily determine the values numerically. The corresponding solutions are given by
$$
\psi_t(\theta)=\frac{1}{N} \matze\left(a_t;-4 t ;\frac{\theta}{2}\right)
$$
where $\matze$ denotes the lowest order first kind solution of the ordinary Mathieu equation, and $N$ is again the normalisation factor. We note that the Fourier coefficients $\widehat\psi_t$ are explicit functions of $a_t$ and $t$, determined by the recurrence relations. Up to second order, we have
\begin{align*}
\psi_t(\theta)&=\frac 1N \Big(1-\frac {a_t}{4t}\cos\theta +\left[\frac{(a_t-4)a_t}{16t^2}+8t\right]\cos 2\theta
+O(\cos 3\theta)\, \Big)
\end{align*}
The relevance of the Mathieu functions in the context of circular uncertainty relations has been noted e.g. in Ref.~\onlinecite{rehacek}.

\subsection{Discrete metric case}\label{sec:discrete}
With $d_{\mathbb Z}=\ddi$, we have the eigenvalue equation
\begin{equation}\label{discrev}
  \bigl(\id -|\phi_0\rangle\langle \phi_0| +t V\bigr)\psi =\lambda \psi,
\end{equation}
with $V$ as in the previous section, and $\phi_0(\theta)=1/\sqrt{2\pi}$ the constant function. This allows us to solve for  $\psi$:
\begin{equation}\label{psidisc}
  \psi_t(\theta) = \frac A{1-\lambda+t V(\theta)}\ ,
\end{equation}
where $A>0$ is  the normalization constant. Inserting this into \eqref{discrev} gives the consistency condition
\begin{equation}\label{disconcis}
  \frac{1}{2\pi}\int_{-\pi}^\pi \frac{d\theta}{1-\lambda +t V(\theta)}=1.
\end{equation}
The smallest positive solution $\lambda=c(t)$ of this equation will give us the desired bound.

However, we can also proceed more directly by using just the functional form \eqref{psidisc}, which we can further simplify to the one-parameter family
$\psi(\theta)=A/(\mu+ V(\theta))$, with a single parameter $\mu$. We then have to solve three integrals:
\begin{align*}
 I_1(\mu)&=\displaystyle\int_{-\pi}^\pi \frac{d\theta}{\mu + V(\theta)} \\
 I_2(\mu)&=\displaystyle\int_{-\pi}^\pi \frac{d\theta}{\bigl(\mu + V(\theta)\bigr)^2}            =\displaystyle-\frac{dI_1(\mu)}{d\mu}\\
 I_3(\mu)&=\displaystyle\int_{-\pi}^\pi \frac{d\theta\ V(\theta)}{\bigl(\mu + V(\theta)\bigr)^2} =\displaystyle I_1(\mu)-\mu I_2(\mu).
\end{align*}
Then the pair of variances
\begin{eqnarray}\label{gendis}
  \Delta_\mdi(N)&=& 1-\frac{I_1(\mu)^2}{2\pi I_2(\mu)} \\
  \Delta^2(\Theta)&=& \frac{I_3(\mu)}{I_2(\mu)}
\end{eqnarray}
lies on the tradeoff curve. One can check that this is consistent with the Legendre transform picture, i.e., condition \eqref{disconcis} in the form
\begin{equation}\label{disconcis1}
  I_1\left(\frac{1-c(t)}t\right)=2\pi t
\end{equation}
and its derivative and the parameter identification $\mu=(1-c)/t$.

Now for the arc metric we have $V(\theta)=\theta^2$ and
\begin{equation}\label{I1arc}
  I_1(\mu)=\frac2{\sqrt\mu}\arctan\bigl(\frac\pi{\sqrt\mu}\bigr).
\end{equation}
Trying to eliminate $\mu$ from \eqref{gendis} leads to a transcendental equation, so one cannot give a closed inequality involving just the variances.

For the chord metric we have
$V(\theta)=2(1-\cos\theta)$ and
\begin{equation}\label{I1cho}
  I_1(\mu)=\frac{2\pi}{\sqrt{\mu(\mu+4)}}.
\end{equation}
In this case one can easily eliminate $\mu$ from \eqref{gendis}, and the tradeoff is explicitly described by equation \eqref{dischord}.

\section{Analytic lower bounds}\label{sec:lower}
In this section we establish a variational method for proving uncertainty relations by applying such bounds for the ground state problem. Of course, variational methods for the ground state problem are well-known. Basically they amount to choosing some good trial state, and evaluating the energy expectation: This will be an upper bound on the ground state energy, and it will be a good one if we have guessed well. However, it is notoriously difficult to find lower bounds on the ground state energy.  The idea for finding such bounds is via the following Lemma:

\begin{lemma}\label{boundlemma}Let $V$ be a $2\pi$-periodic real valued potential, and $H$ the Schr\"odinger operator $H\psi=-\psi''+V\psi$ with ground state energy $E_0(H)$. Consider a twice differentiable periodic function $\phi$, which is everywhere $>0$. Then the ground state energy of $H$ is larger or equal to
\begin{equation}\label{elow}
    E_\phi=\min_\theta\left\{V(\theta)-\frac{\phi''(\theta)}{\phi(\theta)}\right\}.
\end{equation}
\end{lemma}

\begin{proof}\let\wt\widetilde
Let
\begin{equation}\label{vv}
  \wt V(x)=\phi''(x)/\phi(x)+ E_\phi
\end{equation}
and $\wt H$ the Schr\"odinger operator with this potential.
Then $\phi$ is an eigenfunction of $\wt H$ with eigenvalue $E_\phi$, and since it was assumed to be positive, it has no nodes and must hence be the ground state eigenfunction. On the other hand, because $E_\phi\leq V(x)-(\wt V(x)-E_\phi)$, we have $\wt V\leq V$ and hence $\wt H\leq H$. By the Rayleigh-Ritz variational principle \cite{ReedS} this implies the ordering of the ground state energies, i.e., $\wt E=E_0(\wt H)\leq E_0(H)$.
\end{proof}

Finding a $\phi$ which gives a good bound is usually more demanding than finding a good $L^2$-approximant for the ground state, because of the highly discontinuous expression $\phi''/\phi$ and the infimum being taken over the whole interval. In particular, the approximate eigenvectors obtained by other methods may perform poorly, even give negative lower bounds on a manifestly positive operator.

The positivity of $\phi$ may be ensured by setting $\phi(\theta)=\exp f(\theta)$; then one has to minimize $V-(f''+(f')^2))$.

We now consider the combinations of metrics $(\dst,\dar)$ and $(\dst, \dch)$, also comparing the results with existing uncertainty relations found in the literature. This demonstrates how our systematic approach relates to many existing (seemingly ad hoc) uncertainty relations.

One remark should be made concerning the comparison with the literature: The uncertainty measure used for the \numb operator is practically always taken to be the usual standard deviation, which can be different from $\Delta_{\mst}(N)$ since in the latter case the infimum is taken only over the set of integers. In general we have
\[
\Delta_{\mst}(N)^2 \geq \sum_{n\in \mathbb Z} (n-{\rm tr}(\rho{N}))^2 \rho\mby{N}(n),
\]
where the right-hand side is the usual standard deviation, which takes a distribution on the integers as a distribution on the reals, which is supported by the integers. Due to the above inequality any uncertainty tradeoff involving the usual standard deviation also implies the same relation for $\Delta_{\mst}(N)$.

\subsection{Case $(\dst,\dar)$}
The literature on this case begins with the observation that the standard uncertainty relation does not hold and needs to be modified; Judge \cite{judge63, judge64} showed in 1963 that the following tradeoff relation
\begin{equation}\label{jbound}
\Delta_\mst(N)\Delta_\mar(\Theta)\geq c \,\left(1-(3/\pi^2) \Delta_\mar(\Theta)^2\right)
\end{equation}
holds with $c=0.16$, and conjectured the same with $c=1/2$. The conjecture was quickly proved in Ref.~\onlinecite{EvettMahmoud} using the Lagrange multiplier method where $\Delta_\mdi(N)$ is minimised under the constraint of fixed $\Delta_\mar(\Theta)^2$, and in Ref.~\onlinecite{BoMaLe,vanLeuven} by showing that the admissible pairs $(\Delta_\mst(N)^2,\Delta_\mar(\Theta)^2)$ lie above the tangent lines of the curve determined by the equality in \eqref{jbound}. Both methods are essentially equivalent to our approach, and explicitly involve the same eigenvalue problem. In Ref.~\onlinecite{BoMaLe} the bound leading to \eqref{jbound} with the optimal constant was obtained using special properties of the confluent hypergeometric function.

We first show how the above Lemma can easily be applied to derive \eqref{jbound} with the optimal constant $c=1/2$. An essentially identical procedure works also in other cases below. The potential is $V(\theta)=\theta^2$, and we have label the uncertainty pairs as $(x,y)=(\Delta_\mar(\Theta)^2,\Delta_\mst(N)^2)$. As the simplest ansatz we take $f(\theta)=\log\phi(\theta)$ even, hence a polynomial in $\theta^2$, which take as quadratic.  The boundary condition $f'(\pi)=0$ then leaves the one-parameter family
\begin{equation}\label{farc}
   f(\theta) = -\frac a2\, \theta^2\Bigl(1-\frac{\theta^2}{2\pi^2}\Bigr),
\end{equation}
where $a\in{\mathbb R}$ is to be optimised later. The bound given by the Lemma on the uncertainty pair $(x,y)$ is then
\begin{eqnarray}
  y+tx&\geq& E_0\geq E_\phi=\inf_\theta\Bigl\{t\theta^2-f''(\theta)-f'(\theta)^2\Bigr\} \nonumber\\
      &=&a+\inf_\theta \Bigl\{\bigl(t-a^2-\frac{3 a}{\pi ^2}\bigr)\theta^2+\frac{2 a^2}{\pi ^2}\theta^4-\frac{a^2}{\pi ^4} \theta^6 \Bigr\}.\nonumber
\end{eqnarray}
This inequality is valid for any $a>0$ and $t>0$. We choose $t$ so that the linear term in $\xi$ vanishes, i.e., $t=a^2+3a\pi^{-2}$. The remaining polynomial is then positive because of $\theta\leq\pi$, and hence takes its minimum at $\xi=0$. Therefore,
\begin{equation}\label{lowbdarc}
  y+\left(a^2+\frac{3a}{\pi^2}\right)x\geq a.
\end{equation}
The optimal value here is $a=(\pi^2-3x)/(2\pi^2 x)$. Note that this is alwas positive, because the equidistribution has the largest variance, namely $x=\pi^2/3$. Substituting the optimal $a$ in \eqref{lowbdarc} we get exactly \eqref{jbound}.

\subsection{Case $(\dst, \dch)$}
We first recall from \eqref{vMis} that $|\langle e^{i\theta}\rangle_{\rho\mby\Theta}| = 1-\Delta_{\mch}(\Theta)^2/2$. Hence, the von Mises ``circular variance" is associated with the sine and cosine operators $\sin \Theta$ and $\cos \Theta$, introduced by Carruthers, Nieto, Louisell, Susskind, Glogower, and others to study the ``quantum phase problem" \cite{nietorev,carruthers,breitenberger}. The idea was to replace the singular commutator $[N,\Theta]$ by the well-defined relations
\begin{align}\label{sincom}
[N, \sin\Theta] & =i\cos\Theta, & [N,\cos\Theta]=-i\sin\Theta.
\end{align}
Combining the usual Robertson type inequalities associated with these commutators, they obtained (Ref.~\onlinecite[eq. (4.11)]{carruthers}) the tradeoff relation
\begin{equation}\label{CN}
\Delta_{\mst}(N) \Delta_{\mch}(\Theta)\geq \frac{1-\tfrac 12\Delta_\mch(\Theta)^2}{2\bigl[1-
\tfrac14\Delta_\mch(\Theta)^2\bigr]^{1/2}},
\end{equation}
expressed here in quantities relevant for our discussion. It was shown by Jackiw \cite{jackiw} that there are no states for which this inequality is saturated, i.e., this bound is not sharp. It is interesting to note that by replacing the square root term with its trivial upper bound $1$, we get
\begin{equation}\label{judgechord}
\Delta_{\mst}(N) \Delta_{\mch}(\Theta) \geq \frac {1-\tfrac 12\Delta_\mch(\Theta)^2}{2},
\end{equation}
which is just the version of Judge's bound \eqref{jbound} for this metric. Other lower bounds were studied relatively recently \cite{rehacek} by using approximations of the Mathieu functions associated with the exact tradeoff curve.

We first show how \eqref{judgechord} can be obtained using Lemma \ref{boundlemma} by applying the same procedure as above. Interestingly, the relevant trial states are exactly the ones saturating the Robertson inequality for the first commutator in \eqref{sincom}, that is, we take $f(\theta)=a\cos\theta$. Then the resulting variational expression is a function of the variable $\cos\theta$, and hence of the potential $V$:
\begin{equation}\label{lowbdcho}
  y+tx\geq\inf_V\Bigl\{a-\bigl(t-\frac a2-a^2\bigr)V+\frac{a^2}4 V^2 \Bigr\}
\end{equation}
Again it is a good choice to take $t$ so that the first order term in $V$ vanishes, so that the remaining infimum is attained at $V=0$. This gives $t=a/2+a^2$ and
\begin{equation}\label{lowbdcho1}
  y\geq a-\bigl(\frac a2+a^2\bigr)x =\frac{(2-x)^2}{16 x}.
\end{equation}
where at the last equality we have substituted the optimal value $a=(2-x)/(4x)$. On taking the square root this is \eqref{judgechord}.

In order to obtain analytic bounds better than the Carruthers-Nieto tradeoff \eqref{CN}, we apply our method with a trial function which is second order in $\cos\theta$: We take $f(\theta)=a\cos\theta+b (\cos\theta)^2$. The expression to be minimized over $\theta$ can still be written as a polynomial in the potential, and numerical inspection suggests once again that it is a good idea to choose the parameters $t$ and $b$ so that coefficients of $V$ and $V^2$ vanish. This gives linear equations for $t$ and $b$, and the resulting polynomial has its unique minimum, namely $a$, at $V=0$. The analogue of \eqref{lowbdcho1} is then
\begin{equation}\label{lowbdcho2}
  y+\frac{a \left(8 a^2+5 a+2\right)}{8 a+4}\,x\geq a.
\end{equation}
Optimizing $a$ now leads to a third order algebraic equation for which the Cardano solution gives a useless expression in terms of roots. If one just wants the tradeoff curve, the solution is actually not necessary. Defining the coefficient of $x$ as a function $g(a)$, so that $y+g(a)x\geq y$. Optimality requires 
$xg'(a)=1$, so we get the tradeoff curve in parametrized form $a\mapsto (1/g'(a),a-g(a)/g'(a))$.

\section{Outlook}\label{sec:outlook}
The methods employed in this paper for obtaining preparation uncertainty bounds can be applied to a large variety of similar problems. However, the derivation of measurement uncertainty bounds relied entirely on the theorem that phase space symmetry makes the two coincide. It is therefore no surprise that the case of positive number and phase seems much harder to tackle, and the exact uncertainty region is yet to be determined although (non-strict) uncertainty bounds have recently been proven\cite{LPS17}.
Independent efficient methods for obtaining sharp bounds for measurement uncertainty so far have not been found, and it would be highly desirable to find such methods. A possible substitute might be a proof of the conjecture that measurement uncertainty is always larger than preparation uncertainty. Although this inequality must be strict in general, in that way the easily computed preparation uncertainty bounds would automatically be valid (but usually suboptimal) measurement uncertainty bounds. However, the only evidence for supporting such a conjecture is the comparison of cases where either kind of uncertainty vanishes, so such a result is perhaps too much to hope for.

\section*{Acknowledgements}
We thank Joe Renes for suggesting also the discrete metric on $\Z$, and Rainer Hempel for helpful communications concerning the variational principle in Section~\ref{sec:lower}.

RFW acknowledges funding by the DFG through the research training group RTG 1991. JK acknowledges funding from the EPSRC projects EP/J009776/1 and EP/M01634X/1.



%

\end{document}